%% file: pscl_globecom.tex
\documentclass[10pt, conference]{IEEEtran}
\usepackage{times}
\usepackage{cite}
\usepackage[cmex10]{amsmath}
\usepackage{amssymb}
\usepackage{epsfig,verbatim}
\usepackage{bm}
\usepackage{algorithm}
\usepackage{algpseudocode}
\usepackage{graphicx}
\usepackage{xcolor}
\usepackage{tikz}
\usetikzlibrary{arrows,positioning,shapes.geometric,calc}
\usepackage{pgfplots}
\usepackage{url}

\DeclareMathOperator*{\sgn}{sgn}
\DeclareMathOperator{\PM}{PM}
\DeclareMathOperator*{\arctanh}{arctanh}

\newtheorem{MyTheorem}{Theorem}

\begin{document}

\title{Partitioned List Decoding of Polar Codes: Analysis and Improvement of Finite Length Performance}

\author{\IEEEauthorblockN{Seyyed Ali Hashemi\IEEEauthorrefmark{1}, Marco Mondelli\IEEEauthorrefmark{2}, S. Hamed Hassani\IEEEauthorrefmark{3}, R\"udiger Urbanke\IEEEauthorrefmark{4}, and Warren J. Gross\IEEEauthorrefmark{1}}\IEEEauthorblockA{\IEEEauthorrefmark{1}McGill University, Canada; emails: \texttt{seyyed.hashemi@mail.mcgill.ca, warren.gross@mcgill.ca}}\IEEEauthorblockA{\IEEEauthorrefmark{2}Stanford University, USA; email: \texttt{mondelli@stanford.edu}}\IEEEauthorblockA{\IEEEauthorrefmark{3}University of Pennsylvania, USA; email: \texttt{hassani@seas.upenn.edu}}\IEEEauthorblockA{\IEEEauthorrefmark{4}EPFL, Switzerland; email: \texttt{ruediger.urbanke@epfl.ch}}}

\maketitle

\begin{abstract}
Polar codes represent one of the major recent breakthroughs in coding theory and, because of their attractive features, they have been selected for the incoming 5G standard. As such, a lot of attention has been devoted to the development of decoding algorithms with good error performance and efficient hardware implementation. One of the leading candidates in this regard is represented by successive-cancellation list (SCL) decoding. However, its hardware implementation requires a large amount of memory. Recently, a partitioned SCL (PSCL) decoder has been proposed to significantly reduce the memory consumption \cite{hashemi_PSCL}. 

In this paper, we examine the paradigm of PSCL decoding from both theoretical and practical standpoints: (i) by changing the construction of the code, we are able to improve the performance at no additional computational, latency or memory cost, (ii) we present an optimal scheme to allocate cyclic redundancy checks (CRCs), and (iii) we provide an upper bound on the list size that allows MAP performance.
\end{abstract}


%%%%%%%%%%%%%%%%%%%%%%%%
\section{Introduction}\label{sec:intro}

\input{1_intro}

%%%%%%%%%%%%%%%%%%%%%%%%

%%%%%%%%%%%%%%%%%%%%%%%%%%%%
\section{Preliminaries} \label{sec:prel}

\input{2_prel}

%%%%%%%%%%%%%%%%%%%%%%%%%%%%

%%%%%%%%%%%%%%%%%%%%%%%%%%%%
\section{Improved Code Construction} \label{sec:interp}

\input{3_interpolation}

%%%%%%%%%%%%%%%%%%%%%%%%%%%%

%%%%%%%%%%%%%%%%%%%%%%%%%%%%
\section{CRC Selection Scheme} \label{sec:crc}

\input{4_crc}

%%%%%%%%%%%%%%%%%%%%%%%%%%%%

%%%%%%%%%%%%%%%%%%%%%%%%%%%%
\section{List Size Requirement for MAP Decoding}\label{sec:list}

\input{5_SCL-theory}

%%%%%%%%%%%%%%%%%%%%%%%%%%%%

%%%%%%%%%%%%%%%%%%%%%%%%%%%%
\section{Conclusions} \label{sec:concl}

\input{6_conclusions}

%%%%%%%%%%%%%%%%%%%%%%%%%%%%

\section*{Acknowledgment}
The work of M. Mondelli is supported by an Early Postdoc. Mobility fellowship from the Swiss National Science Foundation. This work was done while S. H. Hassani was visiting the Simons Institute for the Theory of Computing, UC Berkeley, as a research fellow. The authors would like to thank Mr. Furkan Ercan of McGill University for providing the LDPC results.

\newcommand{\SortNoop}[1]{}

%%%%%%%%%%%%%%%%%%%%%%%

%%%%%%%%%%%%%%%%%%%%%%%%%%%%%%%%%%%%%%%%%%%%

\end{document}

%% file: 1_intro.tex
Polar codes, introduced by Ar{\i}kan \cite{Ari09}, achieve the capacity of any binary memoryless symmetric (BMS) channel with encoding and decoding complexity $\Theta(N \log_2 N)$, where $N$ is the block length of the code. The code construction can be performed with complexity $\Theta(N)$ \cite{TV13con, RHTT}. Furthermore, by taking advantage of the partial order between the synthetic channels, it is possible to construct a polar code by computing the reliability of a sublinear fraction of the synthetic channels \cite{MHU17}. A unified characterization of the performance of polar codes in several regimes is provided in \cite{MHU15unif-ieeeit}. Here, let us just recall the following basic facts: the error probability scales with the block length roughly as $2^{-\sqrt{N}}$ \cite{ArT09}; the gap to capacity scales with the block length as $N^{-1/\mu}$, and bounds on the scaling exponent $\mu$ are provided in \cite{HAU14, GB14, MHU15unif-ieeeit}; and polar codes are not affected by error floors \cite{MHU15unif-ieeeit}. All the results mentioned above hold under the successive-cancellation (SC) decoder originally proposed by Ar{\i}kan in \cite{Ari09}.

A SC list (SCL) decoder with space complexity $O(L N)$ and time complexity $O(L N \log_2 N)$ is presented in \cite{TVa15}, where $L$ is the size of the list. Empirically, the use of several decoding paths yields an error probability comparable to that under optimal MAP decoding with practical values of the list size. In addition, by adding only a few extra bits of cyclic redundancy check (CRC) precoding, the resulting performance is comparable with state-of-the-art LDPC codes. Because of their attractive properties, polar codes have been selected for the eMBB control channel in the incoming 5G standard \cite{3gpp_polar}.

Although SCL significantly improves the error-correction performance of polar codes, it comes at the cost of lower speed and higher area occupation when implemented in hardware. In particular, it was shown in \cite{balatsoukas} that the high area requirement of SCL is mostly dominated by its memory usage. Several techniques have been proposed to speed up SCL \cite{hashemi_SSCL_TCASI,hashemi_FSSCL} without degrading its error-correction performance. However, all these techniques require the same amount of memory as the SCL decoder. The partitioned SCL (PSCL) algorithm in \cite{hashemi_PSCL} breaks the code into constituent codes (partitions) and each constituent code is decoded with the SCL algorithm. It was shown that the memory requirement of PSCL can be significantly reduced as the number of partitions increases. However, as the number of partitions increases, the error-correction performance of PSCL degrades in comparison with SCL at the same list size.

In this paper, we present several improvements on the original PSCL scheme. These improvements are both practical, in the sense that they boost the error-correction performance of the code, and theoretical, in the sense of the analysis and the understanding of the algorithm. In particular, our contributions can be summarized as follows.

\begin{enumerate}

	\item \emph{Improved Code Construction.} Inspired by \cite{MoHaUr14}, we propose a novel code construction which leads to a remarkable gain in the finite length performance. In particular, the idea is to design the polar code for a ``better'' channel, namely for a channel with a larger SNR compared to the channel that is used for transmission. Furthermore, the performance improvement comes at no additional cost in terms of memory, latency, or operational complexity, since we are simply constructing a polar code for a better channel and do not alter any of the encoding or decoding procedures.
	
	\item \emph{CRC Selection Scheme.} We propose a \emph{successive CRC assignment} strategy to select the number of CRC bits. Our approach is based on the successive minimization of the error probability for each partition. As such, this selection strategy is optimal in the sense that it minimizes the total error probability of the PSCL decoding algorithm.
	
	\item \emph{List Size Requirement for MAP Decoding.} We present an upper bound on the size of the list that ensures optimal MAP performance. The proof holds for the special case of the binary erasure channel (BEC), but numerical simulations suggest that the claim holds for the transmission over general channels. The bound exploits the fact that the information bits tend to cluster at the end of the successive decoding process.  
	
\end{enumerate}

The remainder of the paper is organized as follows. After revising some preliminary material about the various decoding algorithms of polar codes in Section \ref{sec:prel}, the following 3 sections contain the contributions of this work: Section \ref{sec:interp} describes the improved code construction; Section \ref{sec:crc} presents the scheme to select the number of CRC bits; and Section \ref{sec:list} provides the upper bound on the size of the list sufficient to ensure MAP performance. Finally, Section \ref{sec:concl} provides some concluding remarks.

%% file: 2_prel.tex
\subsection{Polar Coding}

A polar code of length $N$ with $K$ information bits and rate $R \triangleq K/N$ is denoted by $\mathcal{P}(N,K)$ and can be constructed by concatenating two polar codes of length $\frac{N}{2}$. This recursive construction can be described with a matrix multiplication as
\begin{equation}
\mathbf{x} = \mathbf{u} \mathbf{B}_N \mathbf{G}^{\otimes n} \text{,} \label{eq:polarEnc}
\end{equation}
where $\mathbf{u} = \{u_0,u_1,\ldots,u_{N-1}\}$ is the sequence of input bits, $\mathbf{x} = \{x_0,x_1,\ldots,x_{N-1}\}$ is the sequence of coded bits, $\mathbf{B}_N$ is the bit-reversal permutation matrix, $\mathbf{G}^{\otimes n}$ is the $n$-th Kronecker product of the polarizing matrix $\mathbf{G} = \bigl[\begin{smallmatrix} 1&0\\ 1&1 \end{smallmatrix} \bigr]$, and $n = \log_2 N$.

The transformation \eqref{eq:polarEnc} is \emph{polarizing} in the sense that the synthetic channels seen by the input bits tend to become either completely noiseless or completely noisy. Thus, the information bits are assigned to the positions of $\mathbf{u}$ corresponding to the best $K$ synthetic channels. The remaining positions of $\mathbf{u}$ are ``frozen'' to predefined values that are known at the decoder.  This set of frozen positions is denoted by $\mathcal{F}$. The resulting $\mathbf{u}$ is coded into $\mathbf{x}$ using \eqref{eq:polarEnc}; then, $\mathbf{x}$ is modulated and transmitted through the channel. Throughout this paper, we consider Binary Phase-Shift Keying (BPSK) modulation.

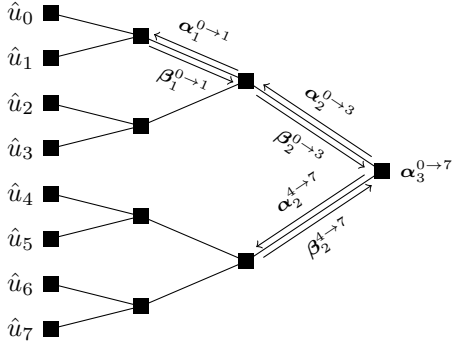
\begin{figure}
\centering
\input{figures/pc-Dec.tikz}
\caption{Polar code SC decoding tree for $N=8$.}
\label{fig:pcDec}
\end{figure}

\subsection{Successive-Cancellation (SC) Decoding}
\label{sec:polar:SCDec}

SC decoding can be represented on a binary tree as shown in Fig. \ref{fig:pcDec} for a polar code of length $8$. The decoder is given as input the vector $\bm{\alpha}_n^{0 \to N-1} = \{\alpha_n^0,\alpha_n^1,\ldots,\alpha_n^{N-1}\}$ of logarithmic likelihood ratio (LLR) values. At each level $t$ of SC decoding, the LLR values are calculated as
\begin{align}
\alpha_t^i &= 2\arctanh \left(\tanh\left(\frac{\alpha_{t+1}^i}{2}\right)\tanh\left(\frac{\alpha_{t+1}^{i+2^{t}}}{2}\right)\right) \text{,} \label{eq:Ffunc} \\
\alpha_t^{i+2^t} &= \alpha_{t+1}^{i+2^{t}} + (1-2\beta_{t}^i)\alpha_{t+1}^i \text{,}
\label{eq:Gfunc}
\end{align}
and the hard-decision estimates $\bm{\beta}_t$ are calculated as
\begin{align}
\beta_t^i &= \beta_{t-1}^{i} \oplus \beta_{t-1}^{i+2^t} \text{,} \label{eq:betaF} \\
\beta_t^{i+2^t} &= \beta_{t-1}^{i+2^t} \text{,} \label{eq:betaG}
\end{align}
where $\oplus$ is the bitwise XOR operation. Eventually, at a leaf node, the $i$-th bit $\hat{u}_i$ is estimated as
\begin{equation}
\hat{u}_i = \beta_{0}^{i} =
  \begin{cases}
    0 \text{,} & \text{if } i \in \mathcal{F} \text{ or } \alpha_{0}^{i}\geq 0\text{,}\\
    1 \text{,} & \text{otherwise.}
  \end{cases} \label{eq:uEst}
\end{equation}
A hardware-friendly formulation of \eqref{eq:Ffunc} can be written as
\begin{equation}
\alpha_t^i = \sgn(\alpha_{t+1}^i)\sgn(\alpha_{t+1}^{i+2^{t}})\min(|\alpha_{t+1}^i|,|\alpha_{t+1}^{i+2^{t}}|) \text{,} \label{eq:FfuncHF}
\end{equation}
which introduces negligible error-correction performance loss in comparison with (\ref{eq:Ffunc}), see \cite{LRSG13}. In this paper, we use (\ref{eq:FfuncHF}) and (\ref{eq:Gfunc}) in the numerical simulations.
Let us consider that the LLR values are quantized with $Q_\alpha$ bits. The memory requirement of SC decoding can be calculated as \cite{LRSG13}
\begin{equation}
M_{\text{SC}} = \left(2N-1\right) Q_\alpha + N-1 \text{.} \label{eq:memSC}
\end{equation}

\subsection{Successive-Cancellation List (SCL) Decoding} \label{sec:polar:SCLDec}

In order to improve the error-correction performance, the SCL decoder with list size $L$, denoted by SCL$(L)$, employs $L$ SC decoders in parallel. For $i\in \{0, \ldots, N-1\}$, when the $i$-th bit $\hat{u}_i$ needs to be estimated, instead of (\ref{eq:uEst}), the SCL decoder creates two paths corresponding to the decisions $\hat{u}_i=0$ and $\hat{u}_i=1$. In order to keep $L$ candidates out of the $2L$ generated paths, a path metric ($\PM$) based on LLRs was developed in \cite{balatsoukas}. A hardware-friendly formulation of such a path metric can be stated as
\begin{align}
&\PM_{{-1}_l} = 0 \text{,} \nonumber \\
&\PM_{{i}_l} = \begin{cases}
    \PM_{{i-1}_l} + |\alpha_{0_l}^i| \text{,} & \text{if } \hat{u}_{i_l} \neq \frac{1-\sgn\left(\alpha_{0_l}^i\right)}{2}\text{,}\\
    \PM_{{i-1}_l} \text{,} & \text{otherwise,}
  \end{cases} \label{eq:PMHF}
\end{align}
where $l$ is the path index.

Eventually, the SCL decoder outputs a list of $L$ candidate codewords. A CRC can be used to assist the SCL decoder in order to find the correct codeword in the list of candidates.
Let us assume the path metric values are quantized with $Q_{\PM}$ bits. The memory requirement of SCL decoding can be derived as \cite{balatsoukas}
\begin{equation}
M_{\text{SCL}} = \left(N+\left(N-1\right)L\right) Q_\alpha + LQ_{\PM} + \left(2N-1\right)L \text{.} \label{eq:memSCL}
\end{equation}

\subsection{Partitioned Successive-Cancellation List (PSCL) Decoding} \label{sec:polar:PSCLDec}

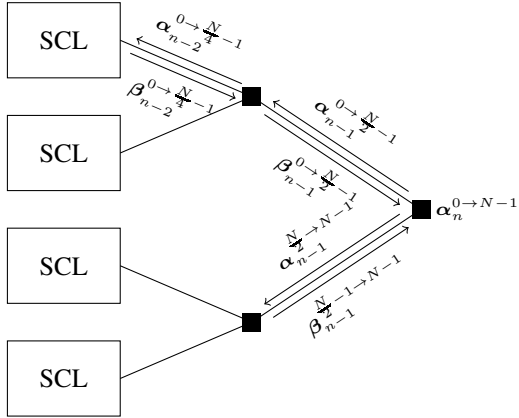
\begin{figure}
\centering
\input{figures/pc-PDec.tikz}
\caption{PSCL tree structure for $P=4$.}
\label{fig:pcPDec} 
\end{figure}

SCL-based decoders have large memory requirements, which results in a large area occupation in the hardware implementation. To overcome this issue, the PSCL decoder was proposed in \cite{hashemi_PSCL}. As shown in Fig.~\ref{fig:pcPDec}, the PSCL decoding tree is divided into $P$ subtrees, called partitions: each partition is a polar code of length $N/P$ and it is decoded by using a CRC-aided SCL decoder with list size $L$. We denote the PSCL decoder with $P$ partitions and list size $L$ as PSCL($P$,$L$).

At the end of each partition, a CRC can be used to select the codeword that is passed to the next partition. In case there is more than one compatible codeword, the one with the largest path metric is selected. As a result, it is not necessary to store $L$ full trees as in SCL, but only $L$ copies of the part of the tree contained in the partitions. Moreover, memory can be shared among the $P$ partitions, which results in significant savings as $P$ increases.

The memory requirement of PSCL with $P$ partitions can be calculated as \cite{hashemi_PSCL}
\begin{align}
M_{\text{PSCL}} = & \left(\sum_{k = 0}^{\log_2P}\frac{N}{2^k} + L\left(\frac{N}{P}-1\right)\right)Q_{\alpha}\nonumber\\
&+ LQ_{\PM} + \sum_{k = 1}^{\log_2P}\frac{N}{2^k} + L\left(\frac{2N}{P}-1\right) \text{,} \label{eq:memPSCL}
\end{align}
where $2 \leq P < N$. It should be noted that in (\ref{eq:memPSCL}), only one copy of the top $\log_2 P$ levels of the tree has to be stored, whereas for the bottom part, $L$ copies are needed. The savings in memory bits guaranteed by the PSCL decoder are summarized in Fig.~\ref{fig:memReq} for $\mathcal{P}(1024,512)$ with $L=8$, $Q_{\alpha} = 6$ and $Q_{\PM} = 8$. Note that the memory requirement of PSCL($P$,$8$) quickly become close to those of the standard SC decoder as $P$ increases. However, this memory saving is obtained at the cost of a higher error probability. Fig.~\ref{fig:FER1kL8comp} shows the degradation in the frame error rate (FER) performance of PSCL, as the number of partitions increases.

\begin{figure}
  \centering
	\input{figures/mem-req}
  \caption{Memory bits required by PSCL($P$,$8$) as a function of the number of partitions $P$. The PSCL decoder provides a significant improvement in the memory requirement compared to the SCL decoder and it approaches the memory required by the SC decoder for practical values of $P$.}
  \label{fig:memReq}
\end{figure}
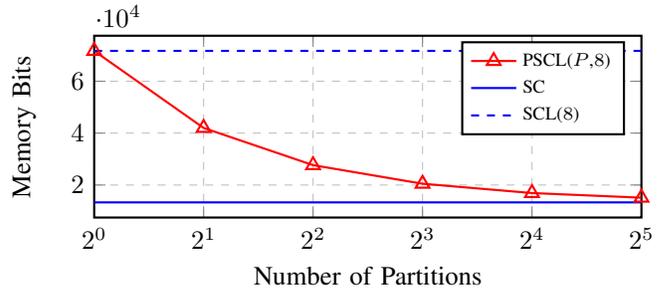

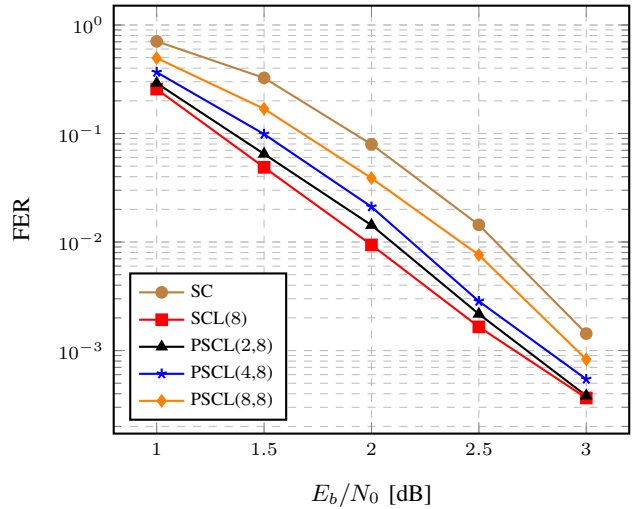
\begin{figure}
  \centering
	\input{figures/fer_comp}
  \caption{FER comparison between SC, SCL(8), and PSCL($P$, 8) for $P\in \{2, 4, 8\}$ when no CRC is used. The polar code is $\mathcal{P}(1024,512)$ and it is optimized for SNR $=2$ dB.}
  \label{fig:FER1kL8comp}
\end{figure}

%% file: figures/pc-Dec.tikz
\begin{tikzpicture}[scale=.8]
\newcommand\Square[1]{+(-#1,-#1) rectangle +(#1,#1)}

% \DoublLine[half of the double line distance]{first node}{second node}{options line 1}{options line 2}{text line 1}{text line 2}
\newcommand\DoubleLine[7][3pt]{%
  \path(#2)--(#3)coordinate[pos=0.9](h1)coordinate[pos=0.1](h2);
  \draw[#4]($(h1)!#1!90:(h2)$)--($(h2)!#1!-90:(h1)$) node [midway, above, sloped] {#6};
  \draw[#5]($(h1)!#1!-90:(h2)$)--($(h2)!#1!90:(h1)$) node [midway, below, sloped] {#7};
}

\fill (5.75,2.625) \Square{.125};

\DoubleLine{5.75,2.625}{3.5,4.125}{<-}{->}{\scriptsize $\bm{\alpha}_2^{0\to 3}$}{\scriptsize $\bm{\beta}_2^{0\to 3}$}
\DoubleLine{5.75,2.625}{3.5,1.125}{<-}{->}{\scriptsize $\bm{\alpha}_2^{4\to 7}$}{\scriptsize $\bm{\beta}_2^{4\to 7}$}

\draw (5.75,2.625) -- (3.5,4.125); 
\draw (5.75,2.625) -- (3.5,1.125);

\fill (3.5,4.125) \Square{.125};
\fill (3.5,1.125) \Square{.125};

\draw (3.5,4.125) -- (1.75,4.875); 
\draw (3.5,4.125) -- (1.75,3.375); 
\draw (3.5,1.125) -- (1.75,1.875); 
\draw (3.5,1.125) -- (1.75,0.375); 

\DoubleLine{3.5,4.125}{1.75,4.875}{<-}{->}{\scriptsize $\bm{\alpha}_1^{0\to 1}$}{\scriptsize $\bm{\beta}_1^{0\to 1}$}

\fill (1.75,4.875) \Square{.125};
\fill (1.75,3.375) \Square{.125};
\fill (1.75,1.875) \Square{.125};
\fill (1.75,0.375) \Square{.125};

\draw (1.75,4.875) -- (.25,5.25);
\draw (1.75,4.875) -- (.25,4.5);
\draw (1.75,3.375) -- (.25,3.75);
\draw (1.75,3.375) -- (.25,3);
\draw (1.75,1.875) -- (.25,2.25);
\draw (1.75,1.875) -- (.25,1.5);
\draw (1.75,0.375) -- (.25,.75);
\draw (1.75,0.375) -- (.25,0);

\fill (.25,5.25) \Square{.125};
\fill (.25,4.5) \Square{.125};
\fill (.25,3.75) \Square{.125};
\fill (.25,3) \Square{.125};
\fill (.25,2.25) \Square{.125};
\fill (.25,1.5) \Square{.125};
\fill (.25,.75) \Square{.125};
\fill (.25,0) \Square{.125};

\node at (-.25,5.25) {$\hat{u}_0$};
\node at (-.25,4.5) {$\hat{u}_1$};
\node at (-.25,3.75) {$\hat{u}_2$};
\node at (-.25,3) {$\hat{u}_3$};
\node at (-.25,2.25) {$\hat{u}_4$};
\node at (-.25,1.5) {$\hat{u}_5$};
\node at (-.25,.75) {$\hat{u}_6$};
\node at (-.25,0) {$\hat{u}_7$};

\node at (6.5,2.625) {\scriptsize $\bm{\alpha}_3^{0\to 7}$};

\end{tikzpicture}

%% file: figures/pc-PDec.tikz
\begin{tikzpicture}[scale=1]
\newcommand\Square[1]{+(-#1,-#1) rectangle +(#1,#1)}

% \DoublLine[half of the double line distance]{first node}{second node}{options line 1}{options line 2}{text line 1}{text line 2}
\newcommand\DoubleLine[7][3pt]{%
  \path(#2)--(#3)coordinate[pos=0.9](h1)coordinate[pos=0.1](h2);
  \draw[#4]($(h1)!#1!90:(h2)$)--($(h2)!#1!-90:(h1)$) node [midway, above, sloped] {#6};
  \draw[#5]($(h1)!#1!-90:(h2)$)--($(h2)!#1!90:(h1)$) node [midway, below, sloped] {#7};
}

\fill (5.75,2.625) \Square{.125};

\DoubleLine{5.75,2.625}{3.5,4.125}{<-}{->}{\scriptsize $\bm{\alpha}_{n-1}^{0\to \frac{N}{2}-1}$}{\scriptsize $\bm{\beta}_{n-1}^{0\to \frac{N}{2}-1}$}
\DoubleLine{5.75,2.625}{3.5,1.125}{<-}{->}{\scriptsize $\bm{\alpha}_{n-1}^{\frac{N}{2}\to N-1}$}{\scriptsize $\bm{\beta}_{n-1}^{\frac{N}{2}-1\to N-1}$}

\draw (5.75,2.625) -- (3.5,4.125); 
\draw (5.75,2.625) -- (3.5,1.125);

\fill (3.5,4.125) \Square{.125};
\fill (3.5,1.125) \Square{.125};

\draw (3.5,4.125) -- (1.75,4.875); 
\draw (3.5,4.125) -- (1.75,3.375); 
\draw (3.5,1.125) -- (1.75,1.875); 
\draw (3.5,1.125) -- (1.75,0.375); 

\DoubleLine{3.5,4.125}{1.75,4.875}{<-}{->}{\scriptsize $\bm{\alpha}_{n-2}^{0\to \frac{N}{4}-1}$}{\scriptsize $\bm{\beta}_{n-2}^{0\to \frac{N}{4}-1}$}

\draw (1.75,4.875)++(-1.5,-.5) rectangle ++(1.5,1);
\draw (1.75,3.375)++(-1.5,-.5) rectangle ++(1.5,1);
\draw (1.75,1.875)++(-1.5,-.5) rectangle ++(1.5,1);
\draw (1.75,0.375)++(-1.5,-.5) rectangle ++(1.5,1);

\node at (1,4.875) {SCL};
\node at (1,3.375) {SCL};
\node at (1,1.875) {SCL};
\node at (1,0.375) {SCL};

\node at (6.5,2.625) {\scriptsize $\bm{\alpha}_n^{0\to N-1}$};

\end{tikzpicture}

%% file: figures/mem-req.tex
\begin{tikzpicture}

\begin{axis}[
	scale = 1,
    xmode=log,
    log basis x=2,
    height=4cm,
    width=\columnwidth,
    xtick=data,
    xlabel={Number of Partitions},
    ylabel={Memory Bits}, ylabel style={yshift=-0.85em},
    legend pos=north east,
    grid=both,
    ymajorgrids=true,
    xmajorgrids=true,
    grid style=dashed,
    xmin = 1,
    xmax = 32,
    thick,
    mark size = 3,
    legend style={
      font=\fontsize{7pt}{7.2}\selectfont,
    },
    legend cell align=left,
]

\addplot[
    color=red,
    mark=triangle,
]
table {
1 71688
2 41992
4 27656
8 20488
16 16904
32 15112
%64 14216
%128 13768
%256 13544
%512 13432
%1024 13376
};
\addlegendentry{PSCL($P$,$8$)}

\addplot[
    color=blue,
]
table {
1 13305
32 13305
};
\addlegendentry{SC}

\addplot[
    color=blue,
    dashed,
]
table {
1 71688
2048 71688
};
\addlegendentry{SCL($8$)}

\end{axis}
\end{tikzpicture}

%% file: figures/fer_comp.tex
\begin{tikzpicture}
  \pgfplotsset{
    label style = {font=\fontsize{9pt}{7.2}\selectfont},
    tick label style = {font=\fontsize{7pt}{7.2}\selectfont}
  }

\begin{axis}[
	scale = 1,
    ymode=log,
    xlabel={$E_b/N_0$ [\text{dB}]},
    ylabel={FER},
    grid=both,
    ymajorgrids=true,
    xmajorgrids=true,
    grid style=dashed,
    thick,
    legend pos=south west,
    legend cell align=left,
]

\addplot[
color=brown,
mark=*,
]
table {
1 0.7068
1.5 0.3253
2 0.0795
2.5 0.0144
3 0.00143279
};
\addlegendentry{\scriptsize SC}

\addplot[
color=red,
mark=square*,
]
table {
1 0.2564
1.5 0.0488
2 0.0094
2.5 0.00165229
3 0.000365922
};
\addlegendentry{\scriptsize SCL($8$)}

\addplot[
color=black,
mark=triangle*,
]
table {
1 0.291
1.5 0.0648
2 0.0143
2.5 0.00217174
3 0.000383349
};
\addlegendentry{\scriptsize PSCL($2$,$8$)}

\addplot[
color=blue,
mark=star,
]
table {
1 0.3669
1.5 0.0987
2 0.0211
2.5 0.00284503
3 0.000544662
};
\addlegendentry{\scriptsize PSCL($4$,$8$)}

\addplot[
color=orange,
mark=diamond*,
]
table {
1 0.4958
1.5 0.1689
2 0.0388
2.5 0.0076
3 0.000829696
};
\addlegendentry{\scriptsize PSCL($8$,$8$)}

\end{axis}
\end{tikzpicture}

%% file: 3_interpolation.tex
In this section, we describe an improved code construction that allows to obtain a gain in the error-correction performance at no additional complexity cost. 

Denote by $W ={\rm BAWGN(SNR}^*)$ the Binary Additive White Gaussian Noise channel with SNR equal to ${\rm SNR}^*$ and let $\mathcal{P}_{\alpha}(N, K)$ be the polar code of length $N$ and rate $K/N$ designed for the  transmission over $W_\alpha={\rm BAWGN(SNR}^* /\alpha)$. In words, when $\alpha=1$, $\mathcal{P}_{\alpha}(N, K)$ is the polar code designed for $W$ and, as $\alpha$ goes from $1$ to $0$, $\mathcal P_\alpha(N, K)$ is a polar code designed for better and better channels.

Consider the transmission of the family of polar codes $\{\mathcal P_\alpha(N, K) : \alpha \in [0, 1]\}$ over the channel $W$. Note that the transmission channel is fixed, while the codes of the family are designed for different channels as $\alpha$ varies. Empirically, one observes that the error probability under MAP decoding decreases as $\alpha$ goes from $1$ to $0$. However, the error probability under SC decoding increases as $\alpha$ goes from $1$ to $0$. The latter is due to the fact that the frozen positions of a polar code are chosen in order to minimize its error probability under SC decoding. Starting from these two observations, in \cite{MoHaUr14} a trade-off between complexity and performance is developed by considering practical decoders (e.g., SCL). The trade-off comes from the fact that, as the decoder becomes more complex (e.g., the list size increases), the best performance is achieved for smaller values of $\alpha$.

\begin{figure}
  \centering
	\input{figures/fer_p2}
	\input{figures/fer_p4}
	\\
	\ref{legend1kL8P2P4}
  \caption{FER for the transmission of $\mathcal{P}(1024,512)$ under PSCL decoding with $L=8$ and $P\in \{2, 4\}$. Different curves correspond to different design SNR values and no CRC is used.}
  \label{fig:FER1kL8P2P4}
\end{figure}
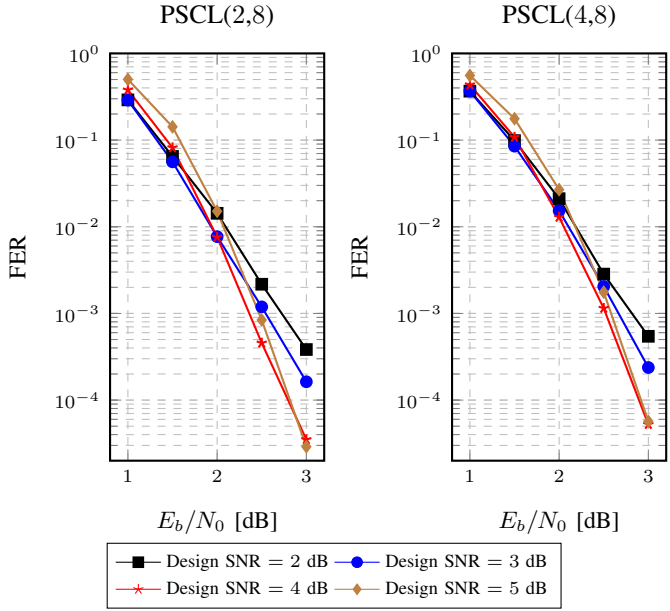

This principle can also be applied to PSCL decoding, as shown in the numerical simulations of Figs. \ref{fig:FER1kL8P2P4}--\ref{fig:FER1kL8comp_arikan}. In particular, different curves of Fig.~\ref{fig:FER1kL8P2P4} refer to different codes (i.e., codes designed for different SNRs). The x-axis corresponds to  SNR of the transmission channel, and y-axis labels the correspondent FER. Note that, at $E_b/N_0 = 3$ dB, a polar code constructed for ${\rm SNR}=5$ dB significantly outperforms the polar code optimized for ${\rm SNR} =3$ dB (which is the SNR of the channel over which the transmission takes place).

Fig.~\ref{fig:FERSNR1kL8P2} plots the error-correction performance of the code as a function of the design SNR. The transmission channel is fixed and it has $E_b/N_0 = 3$ dB. Different curves correspond to different decoding algorithms. For SC decoding, the best error-correction performance is achieved for a code constructed for ${\rm SNR} = 3$ dB. Again, this is to be expected, since the polar code is constructed in order to minimize the error probability under SC decoding. For SCL$(8)$, the best error-correction performance is achieved when the code is constructed for ${\rm SNR} = 5$ dB. For PSCL$(2, 8)$, the optimal design SNR is $4.5$ dB, and for PSCL$(4, 8)$, it is $4$ dB. Note that, as the number of partitions increases, the optimal design SNR for PSCL decoding moves from that for SCL decoding to that for SC decoding. This is due to the fact that, as the number of partitions increases, the error-correction performance of PSCL moves from that of SCL to that of SC (see also Fig.~\ref{fig:FER1kL8comp}). 

Fig.~\ref{fig:FER1kL8comp_arikan} summarizes the gains in the error-correction performance guaranteed by the improved code construction. It also compares the error-correction performance of polar codes with that of an LDPC code of length $1152$ and rate $1/2$ which is used in the WiMAX standard. The dashed curves refer to polar codes in which the design SNR is equal to the SNR of the transmission channel, i.e., polar codes constructed in the standard way. The continuous curves refer to polar codes in which the design SNR is chosen in order to minimize the FER. Both PSCL$(2, 8)$ and PSCL$(4, 8)$ exhibit gains of $0.5$ dB at a target FER of $10^{-3}$. Note that these gains come ``for free'', since changing the design SNR does not affect the computational complexity, the latency or the memory requirement.

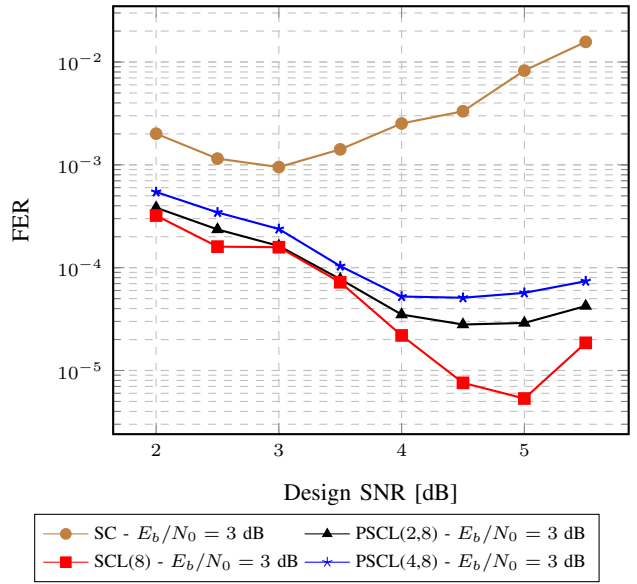
\begin{figure}
  \centering
	\input{figures/fer_snr_p2}
	\\
	\ref{perf-legend1kp2}
  \caption{FER for PSCL($2$,$8$), PSCL($4$,$8$), SCL($8$), and SC decoding of $\mathcal{P}(1024,512)$ as a function of the design SNR. Different curves correspond to different decoding algorithms and no CRC is used. Note that the transmission takes place over a BAWGN with ${\rm SNR} = 3$ dB.}
  \label{fig:FERSNR1kL8P2}
\end{figure}

\begin{figure}
  \centering
	\input{figures/fer_comp_arikan}
  \caption{FER performance comparison between the standard and the improved construction of the code $\mathcal{P}(1024,512)$. Different curves correspond to different decoding algorithms and no CRC is used. Note that gains of $0.5$ dB are obtained at a target FER of $10^{-3}$ for PSCL$(2, 8)$ and PSCL$(4, 8)$. The error-correction performance of the WiMAX LDPC code of length $1152$ and rate $1/2$ is also plotted for comparison.}
  \label{fig:FER1kL8comp_arikan}
\end{figure}
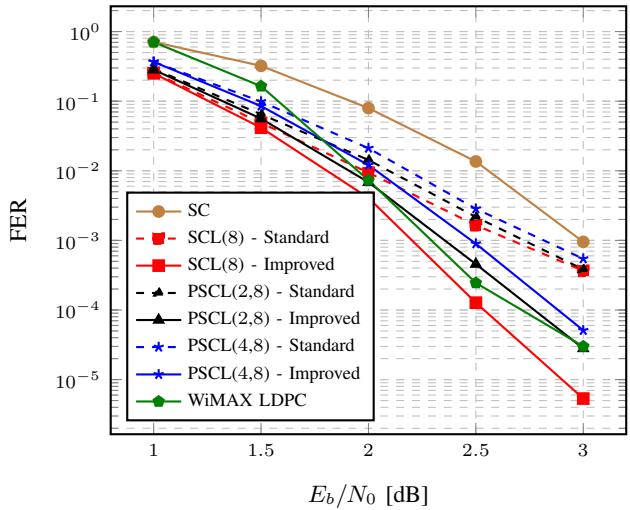

%% file: figures/fer_p2.tex
\begin{tikzpicture}
  \pgfplotsset{
    label style = {font=\fontsize{9pt}{7.2}\selectfont},
    tick label style = {font=\fontsize{7pt}{7.2}\selectfont}
  }

\begin{axis}[
	scale = 1,
    ymode=log,
    xlabel={$E_b/N_0$ [\text{dB}]},
    ylabel={FER},
    ymin=2e-5,
    ymax=1,
    grid=both,
    ymajorgrids=true,
    xmajorgrids=true,
    grid style=dashed,
    thick,
    legend pos=south west,
    title={PSCL($2$,$8$)},
    width=0.5\columnwidth, height=7.0cm,
    legend to name=legend1kL8P2P4,
    legend columns=2,
    legend cell align=left,
]

\addplot[
color=black,
mark=square*,
]
table {
1 0.291
1.5 0.0648
2 0.0143
2.5 0.00217174
3 0.000383349
};
\addlegendentry{\scriptsize Design SNR $=2$ dB}

\addplot[
color=blue,
mark=*,
]
table {
1 0.2892
1.5 0.0561
2 0.0076929
2.5 0.00119049
3 0.000162712
};
\addlegendentry{\scriptsize Design SNR $=3$ dB}

\addplot[
color=red,
mark=star,
]
table {
1 0.3784
1.5 0.0817
2 0.00760681
2.5 0.000455651
3 3.5e-5
};
\addlegendentry{\scriptsize Design SNR $=4$ dB}

\addplot[
color=brown,
mark=diamond*,
]
table {
1 0.5016
1.5 0.1423
2 0.0151
2.5 0.0008398
3 2.9e-5
};
\addlegendentry{\scriptsize Design SNR $=5$ dB}

\end{axis}
\end{tikzpicture}

%% file: figures/fer_p4.tex
\begin{tikzpicture}
  \pgfplotsset{
    label style = {font=\fontsize{9pt}{7.2}\selectfont},
    tick label style = {font=\fontsize{7pt}{7.2}\selectfont}
  }

\begin{axis}[
	scale = 1,
    ymode=log,
    xlabel={$E_b/N_0$ [\text{dB}]},
    ylabel={FER},
    ymin=2e-5,
    ymax=1,
    grid=both,
    ymajorgrids=true,
    xmajorgrids=true,
    grid style=dashed,
    thick,
    legend pos=south west,
    title={PSCL($4$,$8$)},
    width=0.5\columnwidth, height=7.0cm,
    legend cell align=left,
]

\addplot[
color=black,
mark=square*,
]
table {
1 0.3669
1.5 0.0987
2 0.0211
2.5 0.00284503
3 0.000544662
};
%\addlegendentry{SNR $=2$ dB}

\addplot[
color=blue,
mark=*,
]
table {
1 0.368
1.5 0.0851
2 0.0153
2.5 0.00204532
3 0.000237736
};
%\addlegendentry{SNR $=3$ dB}

\addplot[
color=red,
mark=star,
]
table {
1 0.434
1.5 0.1082
2 0.0129
2.5 0.00114612
3 5.23939e-05
};
%\addlegendentry{SNR $=4$ dB}

\addplot[
color=brown,
mark=diamond*,
]
table {
1 0.5573
1.5 0.1766
2 0.0268
2.5 0.00176066
3 5.6874e-05
};
%\addlegendentry{SNR $=5$ dB}

\end{axis}
\end{tikzpicture}

%% file: figures/fer_snr_p2.tex
\begin{tikzpicture}
  \pgfplotsset{
    label style = {font=\fontsize{9pt}{7.2}\selectfont},
    tick label style = {font=\fontsize{7pt}{7.2}\selectfont}
  }

\begin{axis}[
	scale = 1,
    ymode=log,
    xlabel={Design SNR [\text{dB}]},
    ylabel={FER},
    grid=both,
    ymajorgrids=true,
    xmajorgrids=true,
    grid style=dashed,
    thick,
	legend to name=perf-legend1kp2,
    legend columns=2,
    legend cell align=left,
]

%\addplot
%table {
%2 0.0813
%2.5 0.0812
%3 0.093
%3.5 0.1279
%4 0.1674
%4.5 0.2101
%5 0.2733
%5.5 0.3723
%};
%\addlegendentry{\scriptsize SC - $E_b/N_0 = 2$ dB}
%
%\addplot
%table {
%2 0.0143
%2.5 0.00874585
%3 0.0076929
%3.5 0.00685166
%4 0.00760681
%4.5 0.00773575
%5 0.0151
%5.5 0.0267
%};
%\addlegendentry{\scriptsize PSCL($2$,$8$) - $E_b/N_0 = 2$ dB}

\addplot[
color=brown,
mark=*,
]
table {
2 0.002009
2.5 0.00114983
3 0.000951393
3.5 0.00141391
4 0.00252519
4.5 0.00332005
5 0.00825627
5.5 0.0157
};
\addlegendentry{\scriptsize SC - $E_b/N_0 = 3$ dB}

\addplot[
color=black,
mark=triangle*,
]
table {
2 0.000383349
2.5 0.000235479
3 0.000162712
3.5 7.73313e-5
4 3.5e-5
4.5 2.8e-5
5 2.9e-5
5.5 4.23644e-5
};
\addlegendentry{\scriptsize PSCL($2$,$8$) - $E_b/N_0 = 3$ dB}

%\addplot
%table {
%2 0.00878657
%2.5 0.0073346
%3 0.0072
%3.5 0.00407365
%4 0.00416285
%4.5 0.00477145
%5 0.0109
%5.5 0.0221142
%};
%\addlegendentry{\scriptsize SCL($8$) - $E_b/N_0 = 2$ dB}
%
%\addplot
%table {
%2 0.0211
%2.5 0.0142
%3 0.0153
%3.5 0.0121
%4 0.0129
%4.5 0.0142
%5 0.0268
%5.5 0.0369
%};
%\addlegendentry{\scriptsize PSCL($4$,$8$) - $E_b/N_0 = 2$ dB}

\addplot[
color=red,
mark=square*,
]
table {
2 0.000321263
2.5 0.000160157
3 0.000158153
3.5 7.21165e-05
4 2.1924e-05
4.5 7.58389e-06
5 5.34116e-06
5.5 1.85833e-05
};
\addlegendentry{\scriptsize SCL($8$) - $E_b/N_0 = 3$ dB}

\addplot[
color=blue,
mark=star,
]
table {
2 0.000544662
2.5 0.000344578
3 0.000237736
3.5 0.000103619
4 5.23939e-05
4.5 5.10014e-05
5 5.6874e-05
5.5 7.38657e-05
};
\addlegendentry{\scriptsize PSCL($4$,$8$) - $E_b/N_0 = 3$ dB}

\end{axis}
\end{tikzpicture}

%% file: figures/fer_comp_arikan.tex
\begin{tikzpicture}
  \pgfplotsset{
    label style = {font=\fontsize{9pt}{7.2}\selectfont},
    tick label style = {font=\fontsize{7pt}{7.2}\selectfont}
  }

\begin{axis}[
	scale = 1,
    ymode=log,
    xlabel={$E_b/N_0$ [\text{dB}]},
    ylabel={FER},
    grid=both,
    ymajorgrids=true,
    xmajorgrids=true,
    grid style=dashed,
    thick,
    legend pos=south west,
    legend cell align=left,
]

\addplot[
color=brown,
mark=*,
]
table {
1 0.70939
1.5 0.32152
2 0.0795
2.5 0.0136
3 0.000951393
};
\addlegendentry{\scriptsize SC}
% 1 dB
% 1.5 dB
% 2 dB
% 2.5 dB
% 3 dB

\addplot[
dashed,
color=red,
mark=square*,
]
table {
1 0.2564
1.5 0.0488
2 0.0094
2.5 0.00165229
3 0.000365922
};
\addlegendentry{\scriptsize SCL($8$) - Standard}

\addplot[
color=red,
mark=square*,
]
table {
1 0.2511
1.5 0.0416
2 0.00407365
2.5 0.000127306
3 5.34116e-06
};
\addlegendentry{\scriptsize SCL($8$) - Improved}
% 3 dB
% 3.5 dB
% 3.5 dB
% 4.5 dB
% 5 dB

\addplot[
dashed,
color=black,
mark=triangle*,
]
table {
1 0.291
1.5 0.0648
2 0.0143
2.5 0.00217174
3 0.000383349
};
\addlegendentry{\scriptsize PSCL($2$,$8$) - Standard}

\addplot[
color=black,
mark=triangle*,
]
table {
1 0.2827
1.5 0.0561
2 0.00685166
2.5 0.000455651
3 2.8e-5
};
\addlegendentry{\scriptsize PSCL($2$,$8$) - Improved}
% 2.5 dB
% 3 dB
% 3.5 dB
% 4 dB
% 4.5 dB

\addplot[
dashed,
color=blue,
mark=star,
]
table {
1 0.3669
1.5 0.0987
2 0.0211
2.5 0.00284503
3 0.000544662
};
\addlegendentry{\scriptsize PSCL($4$,$8$) - Standard}

\addplot[
color=blue,
mark=star,
]
table {
1 0.3669
1.5 0.0851
2 0.0121
2.5 0.000895103
3 5.10014e-05
};
\addlegendentry{\scriptsize PSCL($4$,$8$) - Improved}
% 2 dB
% 3 dB
% 3.5 dB
% 4.5 dB
% 4.5 dB

\addplot[
color=green!50!black,
mark=pentagon*,
]
table {
1.00 7.06e-01
1.50 1.64e-01
2.00 7.32e-03
2.50 2.45e-04
3.00 3.00e-05
};
\addlegendentry{\scriptsize WiMAX LDPC}

\end{axis}
\end{tikzpicture}

%% file: 4_crc.tex
In this section, we discuss the use of CRC bits and we present a strategy to choose the length of the CRC that minimizes the FER. 

The error-correction performance under SCL decoding is lower bounded by that under MAP decoding. However, in scenarios of interest in practical applications, even the performance of the MAP decoder is not satisfactory compared to state-of-the-art coding schemes, such as LDPC codes. In order to address this issue, it was shown in \cite{TVa15} that, by adding a CRC, the error-correction performance under SCL decoding significantly outperforms that under MAP decoding with no CRC, and it is comparable to that of state-of-the-art LDPC codes.

Denote by $C$ the length of the CRC. In order to send $K$ bits of information for $N$ channel uses, we need to use a code of block length $N$ and rate $(K+C)/N$. More specifically, let us focus on the PSCL decoder with $P$ partitions. Consider the vectors $\mathbf{c} = \{c_1,c_2,\ldots,c_{P}\}$ and $\mathbf{k} = \{k_1,k_2,\ldots,k_{P}\}$, where $c_p$ is the length of the CRC concatenated to the $p$-th partition and $k_p$ represents the number of information bits associated to the $p$-th partition. Clearly,
\begin{equation}
\sum_{p = 1}^{P} c_p = C,\qquad\quad \sum_{p = 1}^{P} k_p = K\text{.}
\label{eq:kSum}
\end{equation}
Then, in the $p$-th partition, we select the most reliable $k_p+c_p$ bits and use the first $k_p$ to store the information bits and the last $c_p$ to store the CRC bits.

On the one hand, if $c_p$ is too small, then some of the incorrect paths might pass the CRC and one of these incorrect paths might be transmitted to the next partition, thus causing an error. On the other hand, if $c_p$ is too large, then the positions chosen to store the CRC might not be sufficiently reliable. Hence, there is a trade-off in the choice of the vector $\mathbf{c}$.

We propose a \emph{successive CRC assignment} strategy that can be described as follows. For $p\in \{1, \ldots, P\}$, we evaluate numerically the FER of the $p$-th partition as a function of $c_p$, assuming that the previous partitions were decoded correctly. Then, we choose the value of $c_p$ that minimizes the FER. In this way, the FER of the PSCL algorithm is minimized. 

Fig.~\ref{fig:1kP220dB} and Fig.~\ref{fig:1kP420dB} consider the transmission of the polar code $\mathcal{P}(1024,512)$ over a channel with ${\rm SNR}=2$ dB and represent the error-correction performance of the $P$ partitions of the PSCL decoder as a function of the length of the CRC for $P=2$ and $P=4$, respectively. The design SNR of the code is equal to the SNR of the transmission channel. For PSCL($2$,$8$), partition $1$ and partition $2$ achieve the best error-correction performance with a CRC of length $4$ and $7$ respectively. Therefore, we conclude that PSCL($2$,$8$) has optimal performance when a CRC of length $4$ is concatenated to the first partition and a CRC of length $7$ is concatenated to the second. Similarly, for PSCL($4$,$8$), partitions $1$, $2$, $3$, and $4$ achieve the best error-correction performance with CRCs of length $2$, $4$, $7$, and $4$, respectively. Therefore, PSCL($4$,$8$) has optimal performance with $\mathbf{c} = \{2,4,7,4\}$.

Fig.~\ref{fig:FER1kL8comp_crc} summarizes the gains in the error-correction performance guaranteed by the successive CRC assignment strategy and also provides a comparison with the LDPC code of length $1152$ and rate $1/2$ which is used in the WiMAX standard. In particular, we consider the SCL($8$), PSCL($2$,$8$), and PSCL($4$,$8$) decoders and we compare the successive CRC assignment scheme with the CRC lengths chosen in \cite{hashemi_PSCL}. All the algorithms exhibit gains of about $1/4$ dB at the target FER of $10^{-3}$. It is worth mentioning that, in order to further improve the performance, one can also optimize the CRC generating polynomial for each partition as discussed in \cite{ZLPP17} for SCL decoders.

\begin{figure}
  \centering
	\input{figures/fer_crc_p2}
  \caption{Effect of the CRC length on the FER for PSCL($2$,$8$) decoding of $\mathcal{P}(1024,512)$. Different curves correspond to different partitions. The SNR of the transmission channel and the design SNR of the code are equal to $2$ dB.}
  \label{fig:1kP220dB}
\end{figure}
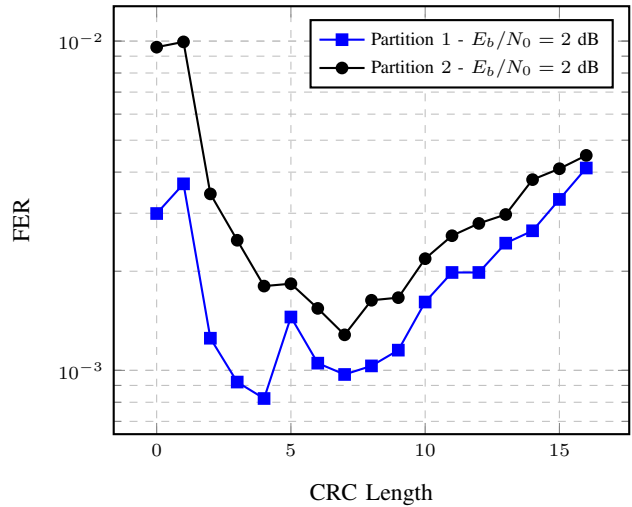

\begin{figure}
  \centering
	\input{figures/fer_crc_p4}
  \caption{Effect of the CRC length on the FER for PSCL($4$,$8$) decoding of $\mathcal{P}(1024,512)$. Different curves correspond to different partitions. The SNR of the transmission channel and the design SNR of the code are equal to $2$ dB.}
  \label{fig:1kP420dB}
\end{figure}
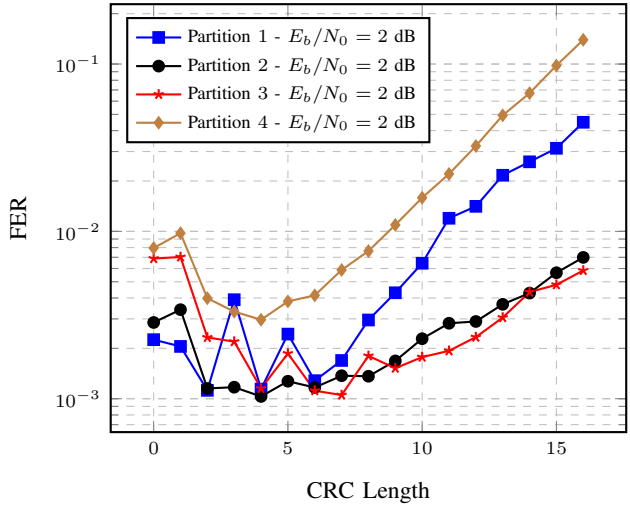

\begin{figure}
  \centering
	\input{figures/fer_comp_crc}
  \caption{FER performance comparison for the transmission of $\mathcal{P}(1024,512)$ between the proposed successive CRC assignment scheme (solid) and the CRC lengths chosen in \cite{hashemi_PSCL} (dashed). Different curves correspond to different decoding algorithms. In the successive CRC assignment scheme, the lengths of the CRCs are optimized separately for each value of the SNR of the transmission channel. Note that gains of about $1/4$ dB are obtained at a target FER of $10^{-3}$ for all the decoding algorithms. The error-correction performance of the WiMAX LDPC code of length $1152$ and rate $1/2$ is also plotted for comparison.}
  \label{fig:FER1kL8comp_crc}
\end{figure}
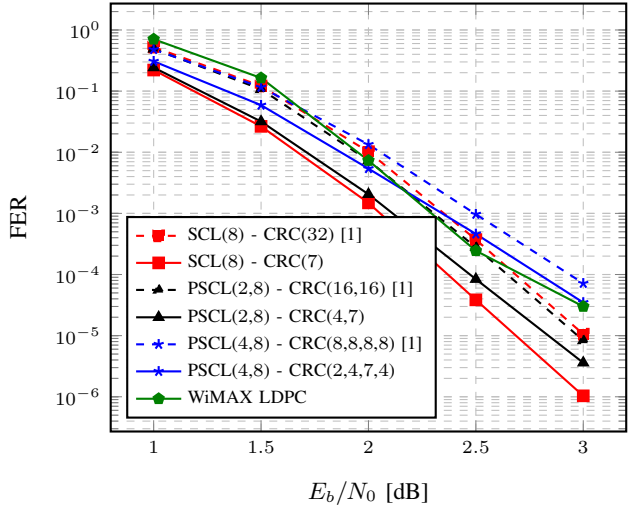

%% file: figures/fer_crc_p2.tex
\begin{tikzpicture}
  \pgfplotsset{
    label style = {font=\fontsize{9pt}{7.2}\selectfont},
    tick label style = {font=\fontsize{7pt}{7.2}\selectfont}
  }

\begin{axis}[
	scale = 1,
    ymode=log,
    xlabel={CRC Length},
    ylabel={FER},
    grid=both,
    ymajorgrids=true,
    xmajorgrids=true,
    grid style=dashed,
    thick,
	legend pos=north east,
	legend cell align=left,
%	legend to name=perf-legend1kp2crc,
%	legend columns=2,
]

\addplot[
color=blue,
mark=square*,
]
table {
0 0.00299479
1 0.0036859
2 0.001252
3 0.000921474
4 0.000821314
5 0.00145232
6 0.00105168
7 0.000971554
8 0.00103165
9 0.00115184
10 0.00161258
11 0.00198317
12 0.00198317
13 0.00243389
14 0.00265425
15 0.00330529
16 0.00411659
%17 0.00424679
%18 0.0054387
%19 0.00542869
%20 0.00625
%21 0.00813301
%22 0.00953526
%23 0.0104167
%24 0.0116386
%25 0.0135116
%26 0.0155549
%27 0.0176583
%28 0.0196815
%29 0.0215645
%30 0.0254006
%31 0.0296474
%32 0.0319712
};
\addlegendentry{\scriptsize Partition $1$ - $E_b/N_0 = 2$ dB}

\addplot[
color=black,
mark=*,
]
table {
0 0.00958534
1 0.00994591
2 0.0034355
3 0.00248397
4 0.00180288
5 0.00183293
6 0.00154247
7 0.00128205
8 0.00163261
9 0.00166266
10 0.00218349
11 0.0025641
12 0.00279447
13 0.00297476
14 0.00379607
15 0.00409655
16 0.0044972
%17 0.00597957
%18 0.00739183
%19 0.0078125
%20 0.0088742
%21 0.0103365
%22 0.0119291
%23 0.0121595
%24 0.0136719
%25 0.0160757
%26 0.0175681
%27 0.0201723
%28 0.0235978
%29 0.0272035
%30 0.028135
%31 0.0320613
%32 0.0354067
};
\addlegendentry{\scriptsize Partition $2$ - $E_b/N_0 = 2$ dB}

%\addplot
%table {
%0 2.22578e-05
%1 2.00321e-05
%2 4.45157e-06
%3 2.82142e-06
%4 1.28411e-06
%5 2.30369e-05
%6 8.41346e-06
%7 3.93557e-07
%8 5.51847e-07
%9 2.33474e-07
%10 3.27857e-07
%11 2.62543e-07
%12 3.13491e-07
%13 6.79053e-07
%14 4.27123e-07
%15 8.34669e-07
%16 7.02879e-07
%%17 7.94923e-07
%%18 1.09465e-06
%%19 1.23655e-06
%%20 1.48386e-06
%%21 1.56968e-06
%%22 3.39526e-06
%%23 6.57302e-06
%%24 4.88587e-06
%%25 3.69011e-06
%%26 7.59836e-06
%%27 5.89178e-06
%%28 8.01282e-06
%%29 1.05168e-05
%%30 1.05432e-05
%%31 2.33707e-05
%%32 2.40385e-05
%};
%\addlegendentry{\scriptsize Partition $1$ - $E_b/N_0 = 3$ dB}
%
%\addplot
%table {
%0 0.00032719
%1 0.0003125
%2 9.01442e-05
%3 4.89516e-05
%4 2.90208e-05
%5 2.10337e-05
%6 8.8141e-06
%7 6.2182e-06
%8 6.1531e-06
%9 2.69662e-06
%10 4.86867e-06
%11 4.72578e-06
%12 4.0597e-06
%13 3.53107e-06
%14 5.85158e-06
%15 5.59228e-06
%16 7.09908e-06
%%17 1.15661e-05
%%18 1.28074e-05
%%19 1.49004e-05
%%20 1.78805e-05
%%21 1.76402e-05
%%22 1.86739e-05
%%23 2.50401e-05
%%24 2.39407e-05
%%25 3.58468e-05
%%26 3.62649e-05
%%27 5.12585e-05
%%28 4.80769e-05
%%29 6.3101e-05
%%30 7.53838e-05
%%31 0.000111289
%%32 0.000103165
%};
%\addlegendentry{\scriptsize Partition $2$ - $E_b/N_0 = 3$ dB}

\end{axis}
\end{tikzpicture}

%% file: figures/fer_crc_p4.tex
\begin{tikzpicture}
  \pgfplotsset{
    label style = {font=\fontsize{9pt}{7.2}\selectfont},
    tick label style = {font=\fontsize{7pt}{7.2}\selectfont}
  }

\begin{axis}[
	scale = 1,
    ymode=log,
    xlabel={CRC Length},
    ylabel={FER},
    grid=both,
    ymajorgrids=true,
    xmajorgrids=true,
    grid style=dashed,
    thick,
	legend pos=north west,
	legend cell align=left,
%	legend to name=perf-legend1kp2crc,
%	legend columns=2,
]

\addplot[
color=blue,
mark=square*,
]
table {
0 0.00225361
1 0.00205329
2 0.00112179
3 0.00390625
4 0.00114183
5 0.00243389
6 0.00128205
7 0.00169271
8 0.00295473
9 0.00429688
10 0.00645032
11 0.0119792
12 0.0141126
13 0.0216446
14 0.0260417
15 0.0313702
16 0.0448718
%17 0.058734
%18 0.0689503
%19 0.0801482
%20 0.09997
%21 0.124269
%22 0.153566
%23 0.170383
%24 0.196384
};
\addlegendentry{\scriptsize Partition $1$ - $E_b/N_0 = 2$ dB}

\addplot[
color=black,
mark=*,
]
table {
0 0.00285457
1 0.00340545
2 0.00115184
3 0.00117187
4 0.00103165
5 0.00127204
6 0.00116687
7 0.0013722
8 0.00136218
9 0.00168269
10 0.00228365
11 0.00282452
12 0.00289463
13 0.00366587
14 0.00427684
15 0.00565905
16 0.00699119
%17 0.00780248
%18 0.0100461
%19 0.0120793
%20 0.0140024
%21 0.0175381
%22 0.020653
%23 0.0244091
%24 0.025
};
\addlegendentry{\scriptsize Partition $2$ - $E_b/N_0 = 2$ dB}

\addplot[
color=red,
mark=star,
]
table {
0 0.00688101
1 0.00703125
2 0.00232372
3 0.00219351
4 0.00115184
5 0.00186298
6 0.00111679
7 0.00105168
8 0.00180288
9 0.00152244
10 0.00177284
11 0.00193309
12 0.00233373
13 0.00305489
14 0.00433694
15 0.00479768
16 0.00583934
%17 0.00739183
%18 0.00843349
%19 0.00922476
%20 0.0116787
%21 0.0135417
%22 0.0169772
%23 0.0207632
%24 0.0228165
};
\addlegendentry{\scriptsize Partition $3$ - $E_b/N_0 = 2$ dB}

\addplot[
color=brown,
mark=diamond*,
]
table {
0 0.00795272
1 0.00975561
2 0.00398638
3 0.0033153
4 0.00296474
5 0.00381611
6 0.00414663
7 0.00587941
8 0.00763221
9 0.0109275
10 0.0158854
11 0.0220453
12 0.0323718
13 0.049349
14 0.0669972
15 0.0979467
16 0.139513
%17 0.1936
%18 0.260417
%19 0.34397
%20 0.433684
%21 0.613241
%22 0.681741
%23 0.671725
%24 0.738311
};
\addlegendentry{\scriptsize Partition $4$ - $E_b/N_0 = 2$ dB}

\end{axis}
\end{tikzpicture}

%% file: figures/fer_comp_crc.tex
\begin{tikzpicture}
  \pgfplotsset{
    label style = {font=\fontsize{9pt}{7.2}\selectfont},
    tick label style = {font=\fontsize{7pt}{7.2}\selectfont}
  }

\begin{axis}[
	scale = 1,
    ymode=log,
    xlabel={$E_b/N_0$ [\text{dB}]},
    ylabel={FER},
    grid=both,
    ymajorgrids=true,
    xmajorgrids=true,
    grid style=dashed,
    thick,
    legend pos=south west,
    legend cell align=left,
]

\addplot[
dashed,
color=red,
mark=square*,
]
table {
1 0.5245
1.5 0.1212
2 0.0102
2.5 0.000372795
3 1.05644e-5
};
\addlegendentry{\scriptsize SCL($8$) - CRC($32$) \cite{hashemi_PSCL}}
%BER
%1 0.175052
%1.5 0.0319305
%2 0.00205938
%2.5 5.59047e-5
%3 1.11215e-6

\addplot[
color=red,
mark=square*,
]
table {
1 0.2212
1.5 0.0265
2 0.00148566
2.5 3.85113e-05
3 1.03624e-6 
};
\addlegendentry{\scriptsize SCL($8$) - CRC($7$)}
%BER
%1 0.0610611	CRC(3)
%1.5 0.00564414	CRC(6)
%2 	CRC(7)
%2.5 4.08581e-6	CRC(11)
%3 9.82001e-8	CRC(12)

\addplot[
dashed,
color=black,
mark=triangle*,
]
table {
1 0.4905
1.5 0.1072
2 0.0073
2.5 0.000278968
3 8.32652e-6
};
\addlegendentry{\scriptsize PSCL($2$,$8$) - CRC($16$,$16$) \cite{hashemi_PSCL}}

\addplot[
color=black,
mark=triangle*,
]
table {
1 0.243
1.5 0.0318
2 0.00205398
2.5 8.38378e-5
3 3.6e-6
};
\addlegendentry{\scriptsize PSCL($2$,$8$) - CRC($4$,$7$)}
%BER
%1 0.0691588	CRC(2,3)
%1.5 0.00706504	CRC(3,4)
%2 0.00205938	CRC(4,7)
%2.5 1.14589e-5	CRC(8,9)
%3 CRC(9,9)	

\addplot[
dashed,
color=blue,
mark=star,
]
table {
1 0.4792
1.5 0.1156
2 0.0134
2.5 0.000966184
3 7.16816e-5
};
\addlegendentry{\scriptsize PSCL($4$,$8$) - CRC($8$,$8$,$8$,$8$) \cite{hashemi_PSCL}}
%BER
%1 0.161502
%1.5 0.0315971
%2 0.00242402
%2.5 0.000110658
%3 6.40934e-6

\addplot[
color=blue,
mark=star,
]
table {
1 0.3075
1.5 0.0589
2 0.0054
2.5 0.000460889
3 3.48187e-5
};
\addlegendentry{\scriptsize PSCL($4$,$8$) - CRC($2$,$4$,$7$,$4$)}
%BER
%1 0.0946043	CRC(2,2,4,2)
%1.5 0.0142195	CRC(2,3,4,3)
%2 0.00205938	CRC(2,4,7,4)
%2.5 5.88534e-5	CRC(6,8,10,6)
%3 2.78686e-6	CRC(7,8,11,6)

\addplot[
color=green!50!black,
mark=pentagon*,
]
table {
1.00 7.06e-01
1.50 1.64e-01
2.00 7.32e-03
2.50 2.45e-04
3.00 3.00e-05
};
\addlegendentry{\scriptsize WiMAX LDPC}

\end{axis}
\end{tikzpicture}

%% file: 5_SCL-theory.tex
In this section, we present a simple upper bound on the list size of the SCL decoder that guarantees the same error-correction performance as the MAP decoder.

Recall that MAP decoding of a code of length $N$ with $K$ information bits requires finding the most reliable codeword out of the $2^{K}$ possible codewords. Note also that SCL decoding with list size $2^K$ provides a list of $2^K$ codewords from which the most reliable one is selected. Therefore, SCL($2^K$) is equivalent to MAP decoding.
For polar codes, in order for SCL to be equivalent to MAP, Theorem~\ref{th:SCLMAP} asserts  that the list size can be chosen to be much smaller than $2^K$. As we will argue shortly, the result of Theorem~\ref{th:SCLMAP} can be particularly useful when low-rate polar codes are deployed. 

\begin{MyTheorem}
Consider the transmission of $\mathcal{P}(N,K)$ over a BEC and let $M$ denote the number of information bits located after the \emph{last} frozen bit. Then, SCL decoding with $L=2^{K-M}$ is equivalent to MAP decoding.
\label{th:SCLMAP}
\end{MyTheorem}

\begin{proof}
Recall that for the case of the transmission over a BEC, the synthetic channels seen by the input bits are also BECs. Hence, the SCL decoder doubles the number of paths every time that one of the information bits cannot be decoded and all these paths are equally likely. Furthermore, when the SCL decoder encounters a frozen position, it tries to decode such a position and, if this is possible, it cancels all the paths that do not agree with the value of this position.

Run the SCL decoder until the last frozen bit. As there are at most $K-M$ information bits until this position, there are at most $2^{K-M}$ equally likely paths. Therefore, the SCL decoder will not exceed its list size. Assume now that the MAP decoder does not fail. This means that there is only one remaining path at the end of the decoding process. Since there are no more frozen bits, no more path cancellations can occur. Therefore, only one path must be available at this point and no new paths will be created while decoding the last $M$ positions. As a result, the SCL decoder also succeeds.
\end{proof}
While we have proved Theorem~\ref{th:SCLMAP} for the BEC, numerical experiments suggest that the result holds for other important channels such as the BAWGN.  

Let us illustrate the importance of this result through the analysis of a practical setting. Polar codes have been selected for the transmission over the control channel in 5G.  Such a task requires codes of short lengths ($\leq 1024$) with rates as low as $1/12$ \cite{3gpp_shortCode}. Consider a polar code of length $128$ and rate $1/12$, which is optimized for SNR $=2$ dB. This code has $10$ information bits and $6$ of them are located after the last frozen bit. Therefore, SCL($16$) can provide the same result as the MAP decoder. This is illustrated in Fig.~\ref{fig:128MAP} when the transmission takes place over the BAWGN channel.

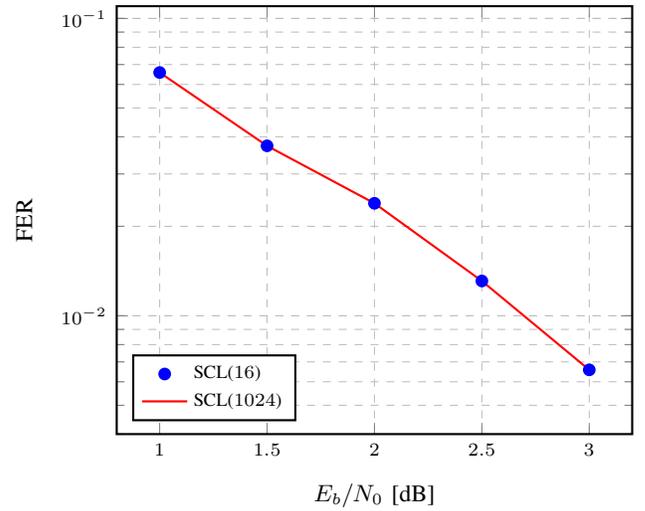
\begin{figure}
  \centering
	\input{figures/fer_comp_MAP}
  \caption{FER performance comparison between SCL($1024$) and SCL($16$) for the transmission of $\mathcal{P}(128,10)$ on a BAWGN channel. Note that, since there are $10$ information bits in the code, SCL($1024$) is equivalent to MAP decoding. It can be seen that, since $6$ of the information bits are located after the last frozen bit, SCL($16$) results in the same error-correction performance as the MAP decoder.}
  \label{fig:128MAP}
\end{figure}

Recall that PSCL results in polar codes of length $N/P$ that are decoded with SCL. Consider the decoding of a polar code of length $1024$ and rate $1/12$ with PSCL with $P=2$. The code is constructed for SNR $=2$ dB. The first partition contains $7$ information bits and $3$ of them are located after the last frozen bit. Therefore, SCL($16$) results in the same error performance as a MAP decoder over that partition. Suppose now that we decode the same code with PSCL with $P=4$. Then, the first partition contains only frozen bits; the second partition contains $7$ information bits and $3$ are located after the last frozen bit; and the third partition contains $9$ information bits and $3$ of them are located after the last frozen bit.
Therefore, for the second and the third partitions, SCL($16$) and SCL($64$) are respectively equivalent to the MAP decoder. In conclusion, while running PSCL decoding, it is possible to perform MAP decoding on some of the partitions with practical values of the list size. 

%% file: figures/fer_comp_MAP.tex
\begin{tikzpicture}
  \pgfplotsset{
    label style = {font=\fontsize{9pt}{7.2}\selectfont},
    tick label style = {font=\fontsize{7pt}{7.2}\selectfont}
  }

\begin{axis}[
	scale = 1,
    ymode=log,
    xlabel={$E_b/N_0$ [\text{dB}]},
    ylabel={FER},
    ymin=0.004,
    ymax=0.11,
    grid=both,
    ymajorgrids=true,
    xmajorgrids=true,
    grid style=dashed,
    thick,
    legend pos=south west,
    legend cell align=left,
]

\addplot[
color=blue,
mark=*,
only marks,
]
table {
1 0.0658
1.5 0.0373
2 0.0239
2.5 0.013095
3 0.00658025
};
\addlegendentry{\scriptsize SCL($16$)}

\addplot[
color=red,
]
table {
1 0.0658
1.5 0.0373
2 0.0239
2.5 0.013095
3 0.00658025
};
\addlegendentry{\scriptsize SCL($1024$)}

\end{axis}
\end{tikzpicture}

%% file: 6_conclusions.tex
Partitioned successive-cancellation list (PSCL) decoding of polar codes has been shown to be a good candidate for applications requiring low memory usage and small area occupation when implemented on hardware. The focus of this paper concerns the finite length error-correction performance of PSCL.
First of all, we present a \emph{novel code construction} inspired by the interpolation method of \cite{MoHaUr14}, which results in a remarkable performance improvement. Then, we propose a \emph{successive CRC selection assignment} scheme to choose the lengths of the CRCs assigned to the different partitions. This approach provides the optimal error-correction performance among all possible CRC assignments. Finally, we prove a simple upper bound on the list size required by the successive-cancellation list (SCL) decoder in order to achieve an error performance equal to that of the optimal MAP decoder. We further argue that in certain practical applications where low-rate codes are used, this upper bound becomes effective in choosing the appropriate list sizes to perform MAP decoding with SCL.

Our analysis also indicates that there is still a gap between the error-correction performance of PSCL and that of SCL in certain high-SNR scenarios. Future work includes reducing this gap by considering more sophisticated PSCL algorithms. 